\documentclass[9pt,final,journal]{IEEEtran}
\usepackage[latin1]{inputenc}
\usepackage{amsmath,amssymb}
\usepackage{graphicx}
\usepackage{enumerate}
\usepackage{cite}
\usepackage{xcolor}

\newtheorem{theorem}{Theorem}[section]
\newtheorem{definition}[theorem]{Definition}
\newtheorem{remark}[theorem]{Remark}
\newtheorem{lemma}[theorem]{Lemma}
\newtheorem{claim}{Claim}

\def\K{\mathcal{K}}
\def\Kinf{\K_{\infty}}
\def\KL{\mathcal{KL}}
\def\R{\mathbb{R}}
\def\N{\mathbb{N}}
\def\X{\mathcal{X}}
\def\E{\mathcal{E}}
\def\U{\mathcal{U}}
\def\S{\mathcal{S}}

\def\dfn{:=}
\def\ms{\medspace}
\def\comp{{\scriptstyle\,\circ}\,}
\def\mer{\hfill $\circ$}

\title{State measurement error-to-state stability results\\ based on approximate discrete-time models} 
\author{A. J. Vallarella and H. Haimovich   
 \thanks{The authors are with Centro Internacional Franco-Argentino de Ciencias de la Informaci\'on y de Sistemas (CIFASIS), CONICET-UNR, Ocampo y Esmeralda, 2000 Rosario, Argentina. {\texttt{\{vallarella,haimovich\}@cifasis-conicet.gov.ar}}}%
\thanks{Work partially supported by ANPCyT grant PICT 2013-0852, Argentina.}%
}

\begin{document}
\maketitle

\begin{abstract}
  Digital controller design for nonlinear systems may be complicated by the fact that an exact discrete-time plant model is not known. One existing approach employs \emph{approximate} discrete-time models for stability analysis and control design, and ensures different types of closed-loop stability properties based on the approximate model and on specific bounds on the mismatch between the exact and approximate models. Although existing conditions for practical stability exist, some of which consider the presence of process disturbances, input-to-state stability with respect to state-measurement errors and based on approximate discrete-time models has not been addressed. In this paper, we thus extend existing results in two main directions: (a) we provide input-to-state stability (ISS)-related results where the input is the state measurement error and (b) our results allow for some specific varying-sampling-rate scenarios. We provide conditions to ensure semiglobal practical ISS, even under some specific forms of varying sampling rate. These conditions employ Lyapunov-like functions. We illustrate the application of our results on numerical examples, where we show that a bounded state-measurement error can cause a semiglobal practically stable system to diverge. 
\end{abstract}


\section{Introduction}
\label{sec:introduction}

Digital controller design for nonlinear systems can be substantially more complicated than for linear systems. One of the main obstacles to the design of adequate digital controllers for nonlinear (continuous-time) plants is that even if a continuous-time plant model is known, the corresponding exact discrete-time model, i.e. the model that describes the state evolution at the sampling instants, can be difficult or even impossible to obtain. This happens because the computation of the exact discrete-time model requires knowledge of the solution to
a nonlinear differential equation. 

Interesting existing results address controller design based on an \emph{approximate} discrete-time model of the plant and ensure the stabilization of the original continuous-time plant in a practical sense
\cite{NesicSCL99,NesicTAC04,NesicLaila_TAC2002,NesicAngeliTAC2002}. 
A general framework for stabilization of disturbance-free sampled-data nonlinear
systems via approximate discrete-time models was developed in \cite{NesicSCL99,NesicTAC04} 
with further generalization for continuous-time plants with disturbances in \cite{NesicLaila_CDC2001,NesicLaila_TAC2002,NesicAngeliTAC2002} 
and for observer design in \cite{Arcak20041931}.
All of these approches are specifically suited to uniform-sampling, i.e. the sampling rate is constant during operation.

In many sampled-data control systems, the use of a non-uniform sampling period is necessary, usually imposed by performance requirements of the plant, hardware limitations or network communication constraints. The occurrence of non-uniform sampling is typical of Networked Control Systems (NCSs). The main feature of NCSs is that the different system components exchange data over a communication network. Results for control design based on approximate discrete-time models also exist for different NCS scenarios. In \cite{van2012discrete}, controller design for nonlinear NCSs with time-varying sampling periods, time-varying delays and packet dropouts is considered, based on approximate discrete-time models constructed for nominal sampling period and delays. All of the results mentioned so far ensure practical stability properties for sufficiently small sampling periods, with ultimate bounds of decreasing size in correspondence with decreasing maximum sampling periods.

Related but conceptually different results exist for dual- or multi-rate sampling, where some constraint exists on how fast measurements can be taken \cite{polushin2004multirate, liu2008input, beikzadeh2015multirate}. These results give closed-loop practical stability warranties and employ approximate discrete-time models in order to predict the state evolution at the control update instants, since new measurements may not be available at each of these instants. Specifically, \cite{beikzadeh2015multirate} proposes a multi-rate sampled-data scheme to stabilize a NCS via output feedback using discrete-time approximations. The results of \cite{beikzadeh2015multirate} as well as those of \cite{van2012discrete} give practical stability and robustness warranties for deviations about a nominal situation.  

Non-uniform sampling is also a feature of event-triggered control, where a triggering condition based on a continuous measurement of system variables (e.g. the system state) determines 
when the control action has to be updated \cite{tabuad_tac07,Heemels2012}.
Strategies that also involve non-uniform sampling but require only 
sampled measurements are those of self-triggered control \cite{velfue_rtss03,Anta2010,Heemels2012}.
Self-triggered control computes both the control action and the next sampling instant at which the control action should be recomputed.
In both event-triggered and self-triggered control, we may thus say that sampling is \emph{controller-driven}, since the controller itself is in charge of computing the next sampling instant.

The objective of this paper is to provide practical stability results based on approximate discrete-time models, and allowing some specific non-uniform sampling scenarios. Our results are, however, independent of the mechanism employed to vary the sampling rate. In this context, the main contribution of the current paper is to provide conditions on the approximate discrete-time plant model and the control law in order to ensure closed-loop semiglobal-practical input-to-state stability (see Section~\ref{sec:preliminaries2} for the precise definition), where the `input' is the state measurement error. 

Our results are novel even for the uniform sampling case. The fact that state measurement errors be considered causes the analysis to become substantially different from that of disturbances affecting the plant dynamics (as in \cite{NesicLaila_TAC2002}). This is because if based on a perturbed state measurement, the control action will itself have some error,
and this in turn will cause some intersample error additional to that introduced by the approximation in the discrete-time model. More specific differences between the given and existing results will be explained along the paper. The  fact that some specific varying-sampling-rate scenarios be covered increases the applicability of our results. Related preliminary results dealing with closed-loop stability possibly under controller-driven sampling (without measurement errors) have been given in \cite{AADECA16,RPIC2017}.

The organization of this paper is as follows. This section ends with a brief summary of the notation employed throughout the paper. In Section~\ref{sec:preliminaries} we state the problem and the required definitions and properties. Our main results are given in Section~\ref{sec:main}. An illustrative example is provided in Section~\ref{sec:example} and concluding remarks are presented in Section~\ref{sec:conclusions}. The appendix contains the proofs of some intermediate technical points.

\textbf{Notation:}
$\R$, $\R_{\ge 0}$, $\N$ 
and $\N_0$ denote the sets of real, nonnegative real, natural and nonnegative integer numbers, respectively. A function $\alpha : \R_{\ge 0} \to \R_{\ge 0}$ is of class-$\K$ (we write $\alpha \in \K$) if it is strictly increasing, continuous and $\alpha(0)=0$. It is class-$\K_{\infty}$ ($\alpha \in \Kinf$) if in addition $\alpha(s) \rightarrow \infty$ as $s \rightarrow \infty$. 
A function $\beta : \R_{\ge 0} \times \R_{\ge 0} \to \R_{\ge 0}$ is of class-$\KL$ ($\beta \in \KL$) if $\beta(\cdot,t)\in \K$ for all $t\ge 0$, and $\beta(s,\cdot)$ is strictly decreasing asymptotically to $0$ for every $s$.
We  denote the identity function by id.
We denote the Euclidean norm of a vector $x \in \R^n$ by $|x|$. Given a set $\X \subset \R^n$ we denote its $\epsilon$ neighborhood as $N(\X,\epsilon):= \{x: \inf_{y\in \X} |x-y|\leq \epsilon \}$. 
We denote an infinite sequence as $\{T_i\}:=\{T_i\}_{i=0}^{\infty}$. For any sequences $\{T_i\} \subset \R_{\ge 0}$ and $\{e_i\} \subset \R^m$, and any $\gamma\in\K$, we take the following conventions: $\sum_{i=0}^{-1} T_i = 0$ and $\gamma(\sup_{0\le i\le -1}|e_i|) = 0$. 
Given a real number $T>0$ we denote by $\Phi(T):=\{ \{T_i\} : \{T_i\} \subset (0,T) \}$ the set of all sequences of real numbers in the open interval $(0,T)$. 
For a given sequence we denote the norm $\|\{x_i\}\|:= \sup_{i\geq0} |x_i|$.

\section{Preliminaries}
\label{sec:preliminaries}

\subsection{Problem statement}
\label{sec:problem_statement}
We consider the nonlinear continuous-time plant
\begin{equation}
\label{eq:cs}
\begin{split}
 \dot{x}&=f(x,u),\quad  x(0)=x_0, \\
\end{split}
\end{equation}
where $x(t) \in \R^n$, $u(t) \in \R^m$ 
are the state and control vectors respectively and $f(0,0)=0$.
As in \cite{NesicSCL99}, the function $f : \R^n \times \R^m \to \R^n$ is assumed to be such that for each initial condition and constant control, there exists
a unique solution on some interval $[0,\tau)$ with $0 < \tau \leq \infty$.

We consider that state measurements become available at time instants $t_k$, $k\in \N_0$, that satisfy $t_0 = 0$ and $t_{k+1} = t_k + T_k$,  where $\{T_k\}_{k=0}^{\infty}$ is the sequence of corresponding sampling periods. As opposed to the uniform sampling case where $T_k = T$ for all $k \in \N_0$, the sampling periods $T_k$ may vary. We refer to  this situation as Varying Sampling Rate (VSR). 
The control signal is assumed to be  piecewise constant  (i.e. zero-order hold is present) such that $u(t) = u(t_k) =: u_k$ for all $t\in [t_k,t_{k+1})$.
We denote the state at the sampling instants by $x_k := x(t_k)$. The control action $u_k$ may depend on the state measurement $\hat x_k = x_k + e_k$, with $e_k$ the state measurement error, and also on the sampling period $T_k$, i.e. we have $u_k = U(\hat x_k, T_k)$. The dependence of the control action at time $t_k$, on the sampling period $T_k = t_{k+1} - t_k$ is possible in the following situations: (a) uniform sampling; (b) controller-driven sampling (such as self-triggered control \cite{velfue_rtss03}, or others \cite{haimoose_nahs12}). In the situation (a), the constant sampling period employed can be known before the controller is implemented. This is the setting in \cite{NesicTAC04,NesicAngeliTAC2002,NesicLaila_TAC2002}. In (b), at each sampling instant $t_k$, the controller may select the next sampling instant $t_{k+1}$ and, thus, the current sampling period $T_k$. As a consequence, knowledge of $T_k$ can be employed in order to compute the current control action $u_k$. However, our results will be still valid in the particular case $u_k = \bar U(\hat x_k)$, where the control law does not explicitly depend on the sampling period. In addition, our results are independent of the way in which the sampling periods $T_k$ may vary over time. Fig.~\ref{fig:diagram} depicts the complete closed-loop system.
\begin{figure}[htb]
  \begin{center}
    \includegraphics[width=75mm,scale=1]{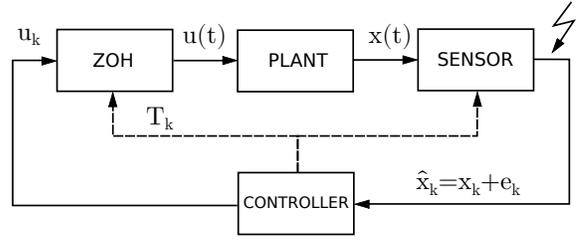} \\
    \caption{Schematic diagram of the controlled plant}
        \label{fig:diagram}
  \end{center}
\end{figure}

The exact discrete-time model for a given nonlinear system is the discrete-time system whose state matches the state of the continuous-time system at every sampling instant. From~(\ref{eq:cs}), and given that the plant input $u$ is held constant over each sampling interval, it is clear that the state value at the next sampling instant, namely $x_{k+1}$, will depend on the current state value $x_k$, the current input value $u_k$, and the sampling period $T_k$. Hence, a discrete-time model for the case of zero-order hold and VSR can be written as
\begin{align}
  \label{eq:fte}
  x_{k+1} &= F^e(x_k,u_k,T_k),
\end{align}
where the superscript $^e$ stands for ``exact''. Due to the fact that the solution to (\ref{eq:cs}) does not necessarily exist for all future times, then the exact discrete-time model $F^e(\cdot,\cdot,\cdot)$ may be not defined for every $(x,u,T) \in \R^n \times \R^m \times \R_{\ge 0}$. However, we assume that $f$ in (\ref{eq:cs}) has sufficient regularity so that for every pair of compact sets $\X \subset \R^n$ and $\U\subset \R^m$, there exists some $T^* > 0$ so that $F^e(x,u,T)$ is defined for every $(x,u,T)\in \X\times\U\times [0,T^*)$. The exact model $F^e$ is in general very hard or even impossible to obtain and thus some type of approximation may be necessary for the design of a stabilizing controller. A superscript $^a$ will be used to denote an approximate model:
\begin{align}
  \label{eq:fta}
  x_{k+1} &= F^a(x_k,u_k,T_k).
\end{align}
We will use `$F$' with no superscript to denote any discrete-time model (either exact or approximate). Under the feedback law $u_k = U(\hat x_k,T_k) = U(x_k + e_k,T_k)$, the closed-loop discrete-time model becomes 
\begin{align}
 \label{eq:system1}
x_{k+1} &= F(x_k,U(x_k+e_k,T_k),T_k) =: \bar F(x_k,e_k,T_k),
\end{align}
where the measurement error $e_k$ (disturbance) is regarded as an input. Hence, given a control law/discrete-time model pair $(U,F)$, the quantity $\bar F$ will denote the composition of $F$ and $U$ as shown in (\ref{eq:system1}). Whenever required, we will write $x(k,\xi,\{e_i\},\{T_i\})$ to denote the solution of (\ref{eq:system1}) at time $k$ with initial condition $\xi$ at $k=0$, and corresponding to the measurement error and sampling period sequences $\{e_i\}$ and $\{T_i\}$.

\subsection{Stability properties for varying sampling rate}
\label{sec:preliminaries2}

Our objective is to provide conditions that ensure stability properties of the exact discrete-time closed-loop system under the varying-sampling-rate case.
The following definition
constitute a natural extension of the 
semiglobal practical ISS \cite{sontag_tac89,Sontag95,NesicLaila_TAC2002} property. 

\begin{definition} 
  \label{def:ISS}
  The system (\ref{eq:system1}) is said to be
    \label{def:SP-ISS-VSR}
    \textit{Semiglobal Practical ISS-VSR} (SP-ISS-VSR) if there exist functions $\beta \in \mathcal{KL}$ and $\gamma \in \K_\infty$ such that for every $M>0$, $E>0$ and $R>0$ there exists $T^\star>0$ such that the solutions of \eqref{eq:system1} satisfy
    \begin{align}
      \label{eq:SP-ISS-VSR}
      |x_k|\leq \beta\left(|x_0|,\sum_{i=0}^{k-1}T_i\right)+\gamma\left(\sup_{0\le i \le k-1}|e_i| \right )+R,
    \end{align}
    for all\footnote{As explained under ``Notation'' in Section~\ref{sec:introduction}, for $k=0$ we interpret $\sum_{i=0}^{-1} T_i = 0$ and $\gamma(\sup_{0\le i \le -1} |e_i| )=0$.} $k\in \N_0$, all $\{T_i\} \in \Phi(T^\star)$, all $|x_0|\leq M$ and all $\|\{e_i\}\|\leq E$. 
    The function $\bar F$ is said to be SP-ISS-VSR if the system~(\ref{eq:system1}) defined by $\bar F$ is SP-ISS-VSR.
\end{definition}

\begin{remark}
  Note that the sequences of constant sampling periods $\{T_i\}$ such that $T_i = T > 0$ for all $i\in\N_0$ are included in the set $\Phi(T^\star)$ whenever $T^\star > T$. This implies that if any of the above properties holds for a certain system
then the respective properties for the uniform sampling case also hold.
\end{remark}

\begin{remark}
  Setting $k=0$ in (\ref{eq:SP-ISS-VSR}), it follows that $|x_0| \le \beta(|x_0|,0) + R$ holds for every $x_0 \in \R^n$ and every $R>0$. As a consequence, any function $\beta\in\KL$ characterizing SP-ISS-VSR has the additional property that $r\le \beta(r,0)$ for all $r\ge 0$. This property will be repeatedly employed for the proof of our results.
\end{remark}

\subsection{Model consistency with non-ideal state measurements}
\label{sec:non-ideal_measurements}

The next definition, as appears in \cite[Definition~1]{Arcak20041931}, ensures that the mismatch between the solutions of the exact and approximate systems over one sampling period is bounded by a value that depends on the sampling period, uniformly over states and inputs in compact sets. This bound tends to zero in a specific way as the sampling period becomes smaller.
%
\begin{definition}
\label{def:consistent} 
The function $F^a$ is said to be consistent with $F^e$
if for each compact set $\Omega \subset \R^n \times \R^m$, there exist
$\rho \in \K$ and $T_0>0$ such that, for all $(x,u)\in \Omega$ and all $T \in [0,T_0]$,
\begin{align}
\label{eq:consistentIneq} 
 |F^e(x,u,T)-F^a(x,u,T)|\leq T \rho(T).
\end{align}
\end{definition}
Definition~\ref{def:consistent} is similar to the definition of one-step consistency in \cite[Defintion~1]{NesicSCL99}. However, the main difference lies in the fact that Definition~\ref{def:consistent} is ``open-loop'' whereas Definition~1 of \cite{NesicSCL99} involves the feedback law. In the current setting of imperfect state knowledge, consideration of the feedback law would be meaningless unless the state measurement error be also considered (cf. Defintion~\ref{def:MSEC} below). 

The following lemma is a minor modification of Lemma~1 of \cite{NesicSCL99} that gives sufficient conditions for the consistency of $F^a$ with $F^e$, based on $f$ and $F^a$, and hence without requiring knowledge of $F^e$. The required modification is due to the fact that in the current setting the feedback law cannot assume perfect knowledge of the state. The proof closely follows the proof of Lemma~1 of \cite{NesicSCL99} but is given in the Appendix for completeness.
%
\begin{lemma}
  \label{lemma:sufcondfor_consistency_f}
  Suppose that
  \begin{enumerate}[i)]
  \item $F^a$ is consistent with $F^{\text{Euler}}$, where $F^{\text{Euler}}(x,u,T) := x + T f(x,u)$, and
    \label{item:lemma_sufcondfor_consistency_f_Euler}
  \item for every pair of compact sets $\X \subset \R^n$ and $\U \subset \R^m$, there exist $\rho \in \K$ and $M>0$ such that for all $(x,u), (y,u) \in \X \times \U$, \label{item:fbounded}
    \begin{enumerate}[a)]
    \item $|f(y,u)|\leq M$, \label{subitem:f_bounded_a}
    \item $|f(y,u)-f(x,u)|\leq \rho(|y-x|)$. \label{subitem:f_Lypschitz_b}
    \end{enumerate}
  \end{enumerate}
  Then, $F^a$ is consistent with $F^e$.
\end{lemma}
The following definition extends the definition of multi-step consistency in \cite{NesicSCL99} to the varying-sampling-rate scenario considered and to the case of state measurement errors. 

\begin{definition}
  \label{def:MSEC}
  The pair $(U,F^a)$ is said to be Multi-Step Error Consistent (MSEC) with $(U,F^e)$ if, for each $L>0$, $\eta>0$ and compact sets $\X \subset \R^n$ and $\E \subset \R^n$, there exist a function $\alpha: \R_{\geq0} \times \R_{\geq0} \rightarrow \R_{\geq0}$ and $T^*>0$ such that
  \begin{align}
    \label{eq:mscee1}
    |\bar F^e(x^e,e,T) - \bar F^a(x^a,e,T)| &\leq \alpha(\delta,T) 
  \end{align}
  for all $x^e,x^a \in \X$ satisfying $|x^e-x^a|\leq \delta$, all $e \in \E$, and all $T\in(0,T^*)$, and for all $\{T_i\} \in \Phi(T^*)$ and $k \in \N$ such that $\sum_{i=0}^{k-1}T_i \leq L$ we have
  \begin{align}
    \label{eq:mscee2}
    \alpha^k(0,\{T_i\}) &:= \overbrace{\alpha(\hdots \alpha(\alpha(}^{k}0,T_0),T_{1}),\hdots,T_{k-1}) \leq \eta.
  \end{align}
\end{definition}
The next lemma constitutes the corresponding extension of \cite[Lemma~2]{NesicSCL99} and shows that under MSEC, the error between the approximate and exact solutions over a fixed time period is reduced by making the maximum sampling period smaller. The proof is given in the Appendix for the sake of completeness.
\begin{lemma}
  \label{LEMA3}
  If $(U,F^a)$ is MSEC with $(U,F^e)$ then for each compact sets $\X \subset \R^n$ and $\E \subset \R^n$, and constants $L>0$ and $\eta>0$, there exists $\tilde T>0$ such that, if $\{T_i\}$, $\{e_i\}$ and $\xi$ satisfy
  \begin{align}
    \label{eq:MSEC1}
    \{T_i\} &\in \Phi(\tilde T), &x^a(k,\xi,\{e_i\},\{T_i\}) &\in \X, & e_k &\in \E, 
  \end{align}
  for all $k$ for which $\sum_{i=0}^{k-1} T_i \in [0,L]$, then, for these values of $k$
  \begin{align}
    \label{eq:MSEC2}
    |\Delta x_k| &:= |x^e(k,\xi,\{e_i\},\{T_i\})-x^a(k,\xi,\{e_i\},\{T_i\})|\leq \eta. 
  \end{align}
\end{lemma}

The next lemma gives a sufficient condition for MSEC. This lemma can be regarded as the corresponding extension of \cite[Lemma 3]{NesicSCL99}.
%
\begin{lemma}
  \label{lem:uFaMSEC}
  If, for each compact sets $\X\subset \R^n$ and $ \E \subset \R^n$ there 
  exist $\rho_0 \in \K$, a nondecreasing function $\sigma: \R_{\geq 0} \rightarrow \R_{\geq 0}$ and $T^* >0$ such that
  for all $T\in(0,T^*)$ and all $x,y \in \X$, $e \in  \E$ we have
    \begin{align}
      \label{eq:FeFaIneq}
      |\bar F^e(x,e,T) - \bar F^a(y,e,T)| &\leq T\rho_0(T) + (1 + T\sigma(T) ) |x-y| 
    \end{align}
  then, $(U,F^a)$ is MSEC with $(U,F^e)$.
\end{lemma}
%
\begin{proof}
  Let $L>0$ and $\eta>0$, and let $\X \subset \R^n$ and $\E \subset \R^n$ be compact sets. Define the compact set $\tilde \X:=N(\X,\eta)$ and let $\tilde\X$  and $\E$ generate $\rho_0\in \K$, $\sigma$ and $T^* >0$. Select $\hat T >0$ sufficiently small so that $\hat T \leq T^*$ and $e^{\sigma(\hat T)L}\rho_0(\hat T)L \leq \eta$. Define
  \begin{equation}
    \label{eq:defalpha}
    \alpha(\delta,T)\dfn T\rho_0(T)+(1+T\sigma(T))\delta.
  \end{equation}
  Let $x,y \in \X$ be such that $|x-y|\le \delta$, let $e\in\E$, and $T\in (0,\hat T)$. From \eqref{eq:FeFaIneq} and (\ref{eq:defalpha}), then \eqref{eq:mscee1} is satisfied. Consider $\{T_i\} \in \Phi(\hat T)$. Using (\ref{eq:defalpha}), we may recursively define $\eta_0 := 0$ and, for $k\ge 1$,
  \begin{align*}
    \eta_k &:= \alpha^k(0,\{T_i\}) =  T_{k-1}\rho_0(T_{k-1})+(1+T_{k-1} \sigma(T_{k-1})) \eta_{k-1}
  \end{align*}
  Recursively employing the latter formula, we may write, for $k\ge 1$ and such that $\sum_{i=0}^{k-1} T_i \le L$,
  \begin{align*}
    \eta_k &= T_{k-1}\rho_0(T_{k-1}) + 
             \sum_{j=0}^{k-2} T_j \rho_0(T_j) \prod_{i=j+1}^{k-1} (1+T_i \sigma(T_i)) \\
           &\leq T_{k-1}\rho_0(T_{k-1}) +
             \sum_{j=0}^{k-2} T_j \rho_0(T_j)  e^{\sigma(\hat T) \sum_{j+1}^{k-1} T_i} \\
           &\leq T_{k-1}\rho_0(T_{k-1}) +
             e^{\sigma(\hat T)L} \sum_{j=0}^{k-2} T_j \rho_0(T_j)  \\
           &\leq e^{\sigma(\hat T)L} \sum_{j=0}^{k-1} T_j \rho_0(T_j) 
             \leq e^{\sigma(\hat T)L}\rho_0(\hat T) \sum_{j=0}^{k-1} T_j  \\
           &\leq e^{\sigma(\hat T)L}\rho_0(\hat T) L \le \eta,
  \end{align*}
  where we have used the facts that for every $a>0$, $1+a \le e^a$, and $T_i < \hat T$ for every $i\in\N_0$.
  This shows that $(U,F^a)$ is MSEC with $(U,F^e)$.
\end{proof}

For the sake of completeness, 
we next state the definition of 
locally uniformly bounded control law as in \cite[Defintion~II.4]{NesicLaila_TAC2002}.
\begin{definition}
$U(x,T)$ is said to be \textit{locally uniformly bounded}
if for every $M>0$ there exist $T^*=T^*(M)>0$ and $C=C(M)>0$ such that
$|U(x,T)|\leq C$ for all $|x| \leq M$ and $T\in(0,T^*)$.
\label{def:localunibounded}
\end{definition}

Lemma~\ref{lemma:sufcondconsistency} gives sufficient conditions for MSEC based on consistency (as per Definition~\ref{def:consistent}) and properties of the control law and approximate model. 
\begin{lemma}
  \label{lemma:sufcondconsistency}
  Suppose that 
  \begin{enumerate}[i)]
  \item $F^a$ is consistent with $F^e$ as per Definition \ref{def:consistent}.
    \label{item:sufcondconsistency_consistentEuler}
 \item $U(x,T)$ is locally uniformly bounded as per Definition \ref{def:localunibounded}.
    \label{item:sufcondconsistency_controllocbound}
  \item \label{item:sufcondconsistency_Fa_nondecreasing_bound}
    For each compact sets $\X, \mathcal{E} \subset \R^n$ there exist a nondecreasing function  $\sigma: \R_{\geq 0} \rightarrow \R_{\geq 0}$ and $T_1>0$ such that for all $x,z\in \X$, $e\in  \mathcal{E}$ and $T\in (0,T_1)$ we have
    \begin{align}
      \label{eq:lemma_sufcondconsistency_Fa_nondecreasing}
      |\bar F^a(x,e,T) - \bar F^a(z,e,T)| &\leq (1+T\sigma (T))|x-z|.
    \end{align} 
  \end{enumerate}
  Then, $(U,F^a)$ is MSEC with $(U,F^e)$.
\end{lemma}
\begin{proof}
  Let $\X \subset \R^n$ and $\E \subset \R^n$ be compact, and let $M>0$ be such that $|x+e|\leq M$ for all $x\in \X$ and $e\in \E$. From~\ref{item:sufcondconsistency_controllocbound}), there exist $T^*(M)>0$ and $C(M)>0$ such that $|U(x+e,T)|\leq C$ for all $T\in(0,T^*)$, $x\in \X$ and $e\in \E$. Define $\mathcal{C}:=\{u \in \R^m : |u|\leq C\}$ and $\Omega := \X \times \mathcal{C}$. From~\ref{item:sufcondconsistency_consistentEuler}), there exist $\rho \in \K$ and $T_0>0$ such that 
  \begin{align}
    \label{eq:lemma28_1}
    |F^e(x,U(x+e,T),T)-F^a(x,U(x+e,T),T)| \leq T\rho(T)
  \end{align}
  for all $T\in(0,T_2)$, $x\in \X$ and $e\in \E$  with $T_2:=\min \{T^*,T_0\}$. For all $x,z \in \X$ and $e\in \E$ we have
  \begin{align}
    \label{eq:lemma28_2}
    |F^e&(x,U(x+e,T),T)-F^a(z,U(z+e,T),T)|\notag \\\leq 
                       &|F^e(x,U(x+e,T),T)-F^a(x,U(x+e,T),T)|  \notag \\
                       &+|F^a(x,U(x+e,T),T)-F^a(z,U(z+e,T),T)|.
  \end{align}
  Defining $\hat T:=\min \{T_1,T_2\}$ and using~(\ref{eq:lemma_sufcondconsistency_Fa_nondecreasing}) and \eqref{eq:lemma28_1} in 
\eqref{eq:lemma28_2}, then
\begin{align}
|F^e(x,U(x+e,T),T)-&F^a(z,U(z+e,T),T)|\notag \\\leq 
&T\rho(T) + (1+T\sigma (T))|x-z|,
\end{align}
hence \eqref{eq:FeFaIneq} holds. By Lemma~\ref{lem:uFaMSEC}, then $(U,F^a)$ is MSEC with $(U,F^e)$.
\end{proof}

In the next section, we give conditions on the approximate model and control law in order to establish the semiglobal practical ISS for the (exact) closed-loop system.

\section{Main Results}
\label{sec:main}

In this section, we show that under MSEC, the SP-ISS-VSR property for the approximate model carries over to the exact model. 

Our first contribution is the following.
\begin{theorem}
  \label{theorem:ISS_aprox_to_exact}
  Suppose that $(U,F^a)$ is MSEC with $(U,F^e)$, and that the system $x^a_{k+1}=\bar F^a(x^a_k,e_k,T_k)$ is SP-ISS-VSR. Then, the system $x^e_{k+1}=\bar F^e(x^e_k,e_k,T_k)$ is SP-ISS-VSR.
\end{theorem}

\begin{IEEEproof}
  Let $\beta\in\KL$ and $\gamma\in\Kinf$ characterize the SP-ISS-VSR property of $x^a_{k+1}=\bar F^a(x^a_k,e_k,T_k)$. Define
\begin{align}
  \label{eq:defHatg}
  \hat \gamma (\cdot)&:=\beta(2\gamma(\cdot),0)+\gamma(\cdot).
\end{align}
Let $M>0$, $E>0$ and $R> 0$, and select $\eta>0$ and $R_a>0$ such that 
\begin{align}
  \label{eq:3etaM}
  6\eta &\leq M,  
  &\beta(6\eta,0)+R_a &<R/2,   
  &R_a &<\eta.
\end{align}
Since $x_{k+1}^a = \bar F^a(x_k^a, e_k, T_k)$ is SP-ISS-VSR, in correspondence with $\tilde M:= M + \gamma(E)$, $E$ and $R_a$, there exists $T^a>0$ such that for all $k \in \N_0$,
\begin{align}
\label{eq:ISS1}
   |x^a_k| \le   \beta \left(|x_{0}^a|,\sum_{i=0}^{k-1}T_i \right)+ \gamma\left(\sup_{0\le i \le k-1}|e_i|\right)+R_a
\end{align}
whenever $|x_{0}^a| \leq \tilde M$ and for all sequences $\{T_i\}$ and $\{e_i\}$ such that $\{T_i\} \in \Phi(T^a)$ and $\|\{e_i\}\| \leq E$. Define
\begin{align}
  \tilde \X&:=\{x \in \R^n :|x|\leq \beta(\tilde M,0)+\gamma(E)+R_a\} \label{eq:kingcrimson}\\
  \E&:=\{e \in \R^n :|e|\leq E\}. 
\end{align}
Let $L>1$ satisfy
\begin{equation}
\label{eq:betaneqeta} 
  \beta\Big(\beta\big(M,0\big)+\gamma(E)+R_a,\ L-1\Big) \leq \eta.
\end{equation}
Take $\tilde T>0$ in correspondence with the tuple $(\tilde \X,\E,L,\eta)$, as per Lemma~\ref{LEMA3}. We will show that $\beta\in\KL$ and $\hat\gamma\in\Kinf$ characterize the SP-ISS-VSR of $x^e_{k+1}=\bar F^e(x^e_k,e_k,T_k)$, with $T^\star:= \min\{1,T^a,\tilde T\}$.
Consider sequences $\{e_i\} \subset \E$ and $\{T_i\} \subset \Phi(T^\star)$. For every $k,\ell\in\N_0$ and $\xi\in\R^n$, define
\begin{multline*}
  \Delta x_{k,\xi}^\ell \dfn\\ |x^e(k-\ell,\xi,\{e_{i+\ell}\},\{T_{i+\ell}\}) - x^a(k-\ell,\xi,\{e_{i+\ell}\},\{T_{i+\ell}\})|.
\end{multline*}
From (\ref{eq:ISS1}) and (\ref{eq:kingcrimson}) it follows that if $|x_{0}^a| \le \tilde M$, then $x_k^a \in \tilde\X$ for all $k \in\N_0$. Consider an initial condition $x_0^e = x_0^a = \xi$ such that $|\xi| \le M \le \tilde M$. 
Since $(U,F^a)$ is MSEC with $(U,F^e)$, by Lemma~\ref{LEMA3} it follows that
for all $k$ such that $\sum_{i=0}^{k-1} T_i \leq L$, we have $|\Delta x_{k,\xi}^0| \leq \eta$. We thus have
\begin{align}
  |x^e(k,&\xi,\{e_i\},\{T_i\})| \leq  |x^a(k,\xi,\{e_{i}\},\{T_{i}\})|+ |\Delta x^0_{k,\xi}| \notag\\
  \label{eq:xeleqxapDelta}
         &\leq  \beta \left(|\xi|,\sum_{i=0}^{k-1}T_i \right)+ \gamma\left(\sup_{0\le i \le k-1}|e_i|\right)+R_a+\eta \\
  \label{eq:bound_1st}
         &\leq \beta \left(|\xi|,\sum_{i=0}^{k-1}T_i \right) 
 + \hat \gamma\left(\sup_{0\le i \le k-1}|e_i|\right)+R,
\end{align}
for all $k$ for which $\sum_{i=0}^{k-1} T_i \leq L$, where we have used the facts that $\gamma\leq\hat \gamma$ and $\eta + R_a \le 6\eta + R_a \le \beta(6\eta,0) + R_a < R/2 < R$. 
For every $k\in\N_0$, define
\begin{align}
s(k):= \max \left\{r\in \N_0 : r \ge k+1, \sum_{i=k}^{r-1} T_i \leq L\right\}.
\end{align}
Note that $s(k) \ge k+1$ for all $k\in\N_0$ because $L>1$ and $T_i < 1$ for all $i\in\N_0$. Also, $\sum_{i=k}^{s(k)-1} T_i > L-T^\star > L-1$ holds for all $k\in\N_0$.
We thus have
\begin{align}
\beta\left (|\xi|,\sum_{i=0}^{s(0)-1} T_i \right ) & \leq \beta(M,L-1) 
\notag \\
&\leq \beta(\beta(M,0)+\gamma(E)+R_a,L-1) 
\leq \eta, \label{eq:Beta_neq_beta_gamma} 
\end{align}
where we have used the fact $M\leq \beta(M,0)$ and (\ref{eq:betaneqeta}).
Evaluating (\ref{eq:xeleqxapDelta}) at $k=s(0)$, using \eqref{eq:3etaM} and $\eqref{eq:Beta_neq_beta_gamma}$, 
and defining \mbox{$x^e_{k}:=x^e(k,\xi,\{e_i\},\{T_i\})$} it follows that 
\begin{align}
|x^e_{s(0)}|
&\leq \beta\left(|\xi|, \sum_{i=0}^{s(0)-1}{T_i}\right) +\gamma\left(\sup_{0\le i \le s(0)-1}|e_i|\right)+R_a+\eta \notag \\
&\leq \gamma\left(E\right)+3\eta 
\leq \tilde M . \label{eq:xeleqgamma2eta} 
\end{align}

\begin{claim}
  \label{clm:iter}
Suppose that $|x^e_{r}|\leq \gamma(E) + 3\eta$ for some $r\in \N_0$. Then,
 $|x^e_{s(r)}| \leq \gamma(E) + 3\eta$ and $|x^e_k| \le \hat\gamma(E) + R$ for all $r \le k \le s(r)$.
\end{claim}

\indent\emph{Proof of Claim~\ref{clm:iter}:}
  Set $x_0^a = x_r^e$ and consider $x_{k+1}^a = \bar F^a(x_k^a,e_{k+r},T_{k+r})$ for all $k\in \N_0$. For all $r\leq k \leq s(r)$, we have
  \begin{align}
    |x^e_k| &= |x^e(k,\xi,\{e_i\},\{T_i\})| = |x^e(k-r,x^e_{r},\{e_{i+r}\},\{T_{i+r}\})| \notag\\
    \label{eq:xexaDelt}
  &\leq |x^a_{k-r}| + |\Delta x_{k,x^e_r}^r|. 
\end{align}
Since $|x_0^a| \le \gamma(E) + 3\eta \le \tilde M$, $\{e_i\} \subset \E$, $\{T_i\} \in \Phi(T^\star) \subset \Phi(T^a)$, and since $\bar F^a$ is SP-ISS-VSR, then for all $k\in\N_0$ we have
\begin{align}
  |x_k^a| &\le \beta\left( |x_0^a| , \sum_{i=0}^{k-1} T_{i+r} \right) + \gamma\left( \sup_{0\le i \le k-1} |e_{i+r}| \right) + R_a\notag\\
  \label{eq:bndxaiter}
  &\le \beta\left( \gamma(E) + 3\eta, \sum_{i=0}^{k-1} T_{i+r} \right) + \gamma(E) + R_a\\
  &\le \beta\left(\tilde M, \sum_{i=0}^{k-1} T_{i+r} \right) +\gamma(E) + R_a.\notag
\end{align}
It thus follows that $x^a_k \in \tilde \X$ for all $k \in \N_0$. By Lemma~\ref{LEMA3}, we have
\begin{align}
  \label{eq:Deltaetabnd}
  |\Delta x_{k,x^e_r}^r| &\leq \eta \quad \text{for all }r\leq k \leq s(r).
\end{align}
Combining the bounds obtained so far, we reach, for all $r \le k \le s(r)$,
\begin{align}  
  |x^e_k| &\leq \beta\left(\tilde M,\sum_{i=r}^{k-1}{T_i}\right) +\gamma\left( E \right) + 2\eta.
  \label{eq:neilyoung}
\end{align}
In particular at $k=s(r)$, and taking into account that $\tilde M= M+\gamma(E) \leq \beta(M,0)+\gamma(E)+R_a$, we obtain
\begin{align}
  |x^e_{s(r)}| 
  &\leq \beta\left(\beta(M,0)+\gamma(E)+R_a,\sum_{i=r}^{s(r)-1}{T_i}\right) +\gamma\left( E \right) +  2\eta \notag \\
  &\leq \gamma\left( E \right)+ 3\eta \label{eq:beatles}
\end{align}
where we have used the fact that $\sum_{i=r}^{s(r)-1} T_i >L-1$ and \eqref{eq:betaneqeta}.

From~(\ref{eq:xexaDelt})--(\ref{eq:Deltaetabnd}), we have, for all $r \le k \le s(r)$,
\begin{align}
  |x^e_k| &\leq \beta(\gamma(E) + 3\eta, 0) + \gamma(E) + R_a + \eta \notag \\
          &\leq \beta(2\gamma(E),0) + \beta(6\eta, 0) + \gamma(E) + R_a + \eta \notag \\
          &\leq \hat\gamma(E) + \beta( 6\eta, 0) + R_a + \eta \notag \\
          &\leq \hat\gamma(E)+ R,  
            \label{eq:last}
\end{align}
where we have used the fact that $\eta < 6\eta < \beta(6\eta,0) + R_a < R/2$. This concludes the proof of the claim.\mer

Since from (\ref{eq:xeleqgamma2eta}) we have $|x^e_{s(0)}| \le \gamma(E) + 3\eta$, iterative application of Claim~\ref{clm:iter} and the fact that $s(k) \ge k + 1$ for all $k\in\N_0$ show that 
\begin{align*}
  |x_k^e| &\le \hat\gamma(E) + R \quad \text{for all }k\ge s(0).
\end{align*}
Combining the latter bound with (\ref{eq:bound_1st}), valid for $0 \le k \le s(0)$, it follows that
\begin{align}
\label{eq:result}
  |x^e_k|\leq \beta\left(|\xi|,\sum_{i=0}^{k-1}T_i\right)+\hat\gamma(E)+R,  \quad \forall k\geq0.
\end{align}
The proof concludes by noting that, by causality, the trajectory $x_k^e$ cannot depend on future values of $e_i$, and hence $E$ in (\ref{eq:result}) can be replaced by $\sup_{0 \le i \le k-1} |e_i|$.
\end{IEEEproof}

Theorem~\ref{theorem:ISS_aprox_to_exact} shows that the SP-ISS-VSR property 
for the exact closed-loop discrete-time model can be ensured under MSEC if 
the approximate closed-loop discrete-time model is SP-ISS-VSR. 
It would thus be useful to have checkable conditions in order to
ensure that the approximate model is suitable. 
A set of conditions is given by Theorem~\ref{lemma:SP-ISS-VSR_Lyap},
which states Lyapunov-based sufficient conditions
for a discrete-time (closed-loop) model to be SP-ISS-VSR.
These conditions consist of specific boundedness and continuity 
requirements
and a Lyapunov-type condition on the closed-loop model. 
Theorem~\ref{lemma:SP-ISS-VSR_Lyap} constitutes the main contribution of the current paper
and its proof, which is highly nontrivial, consitutes our main technical result. 
%
\begin{theorem}
  \label{lemma:SP-ISS-VSR_Lyap}  
  Suppose that \ref{item:SP-ISS-VSR_lemma_3_zero})--\ref{item:SP-ISS-VSR_lemma_4}) are satisfied.
  \begin{enumerate}[i)]  
\item There exists $\mathring T>0$ so that $\bar F(0,0,T) = 0$ for all $T\in(0,\mathring T)$. \label{item:SP-ISS-VSR_lemma_3_zero}
\item There exists $\hat T>0$ such that for every $\epsilon>0$ there exists $\delta=\delta(\epsilon)>0$ 
such that $|\bar F(x,e,T)|<\epsilon$ whenever $|x|\leq \delta$, $|e|\leq \delta$ and $T\in(0,\hat T)$.
  \label{item:SP-ISS-VSR_lemma_1_cont}

\item For every $M\ge 0$ and $E\ge 0$, there exist $C=C(M,E)>0$ and $\check{T} = \check{T}(M,E) > 0$, with $C(\cdot,\cdot)$ nondecreasing in each variable and $\check{T}(\cdot,\cdot)$ nonincreasing in each variable, such that $|\bar F(x,e,T)|\leq C$ for all $|x|\leq M$, $|e|\leq E$ and $T \in (0,\check{T})$. 
\label{item:SP-ISS-VSR_lemma_2_bound}
    \item There exist $\alpha_1, \alpha_2, \alpha_3 \in \K_{\infty}$ and $\rho \in \K$ such that for every $M\ge R>0$ and $E>0$ there exist $\tilde T = \tilde T(M,E,R) > 0$ and $V: \R^n \rightarrow \R_{\geq0}$ such that \label{item:SP-ISS-VSR_d}  \label{item:SP-ISS-VSR_lemma_4}
    \begin{align}
      \label{eq:SP-ISS-VSR_d_1}
      &\alpha_1(|x|) \leq V(x) \leq \alpha_2 (|x|),  \quad \forall x \in \R^n,\text{ and}\\
      \label{eq:SP-ISS-VSR_d_2} 
      V(&\bar F(x,e,T)) -V(x) \leq -T\alpha_3 (|x|), 
      \quad\forall (x,e,T) \text{ satisfying}\notag\\ 
      &\rho(|e|)+ R \le |x| \le M,\quad |e| \le E,\quad T\in (0,\tilde T).
    \end{align} 
  \end{enumerate}
  Then, the system \eqref{eq:system1} is SP-ISS-VSR.
\end{theorem}
\begin{IEEEproof}
We aim to prove that there exist $\beta \in \KL$ and $\gamma \in \K_\infty$ such that for all $M_0>0$, $E_0>0$ and $R_0>0$ there exists $T^\star>0$ such that for all $\{T_i\}\in \Phi(T^\star)$, $|x_0|\leq M_0$, $\|\{e_i\}\|\leq E_0$ and $k\in \N_0$,
the solutions of \eqref{eq:system1} satisfy
\begin{align}
|x_k|\leq \beta\left(|x_0|,\sum_{i=0}^{k-1}T_i\right)+\gamma\left(\sup_{0\le i \le k-1}|e_i| \right )+R_0.
\end{align}
Consider $\rho\in\K$ from \ref{item:SP-ISS-VSR_lemma_4}). Define, $\forall s \geq0$,  $\forall r \geq0$,
\begin{align}
  \X_1(s,r)&:= \{x\in\R^n: |x|\leq \rho(s)+r \} \label{eq:X1} \\
  \E(s)&:= \{e\in\R^n: |e|\leq s \}  \label{eq:E}   \\
  \bar T(s,r) &:= \min\Big\{ \check{T}(\rho(s)+s,s), \check{T}(\rho(r)+r,r), \hat T, \mathring T \Big\} \label{eq:barT}\\
  \label{eq:setS}
  \S(s,r) &:= \X_1(s,r)\times\E(s)\times(0,\bar T(s,r))\\
  \sigma(s,r) &:=\sup_{(x,e,T) \in \S(s,r)} |\bar F(x,e,T)|. \label{eq:1111} 
\end{align}
From \eqref{eq:X1}--\eqref{eq:E}, we have $\X_1(0,0)=\{0\}$ and $\E(0)=\{0\}$. From assumptions~\ref{item:SP-ISS-VSR_lemma_3_zero})--\ref{item:SP-ISS-VSR_lemma_1_cont}), then $\bar T(0,0) > 0$ and $\sigma(0,0)=0$. 
Given that $\check T(\cdot,\cdot)$ is nonincreasing in each variable, note that
$\check T (\rho(s)+r,s) \geq \min \{\check T (\rho(s)+s,s),\check T (\rho(r)+r,r) \}$,
then $\bar T(s,r) \leq \check T (\rho(s)+r,s)$.
Defining $\zeta(s):=\sigma(s,s)$, from \eqref{eq:1111} we have
\begin{align}
\sigma(s,r) &\leq \zeta(s)+\zeta(r). \label{eq:bounding}
\end{align}

\begin{claim}
\label{clm:rightcontinuous}
There exists $\chi \in \K_\infty$ such that $\chi\geq\zeta$.
\end{claim}

\indent\emph{Proof of Claim~\ref{clm:rightcontinuous}:}
Let $\epsilon > 0$ and take $\delta=\delta(\epsilon)$ according to \ref{item:SP-ISS-VSR_lemma_1_cont}). 
Define $\hat \delta:=\min\left\{\frac{\delta}{2},\rho^{-1}(\frac{\delta}{2})\right\}$
(if $\frac{\delta}{2} \notin \text{dom} \ms \rho^{-1}$, just take $\hat\delta=\delta/2$).
Then for all $x\in \X_1(\hat \delta,\hat \delta)$ and $e\in \E(\hat \delta)$ we have $|x|\leq \delta$ and $|e|\leq \delta/2$. 
Given that $|\bar F(x,e,T)| < \epsilon$ for every $|x| \le \delta$, $|e|\leq \delta$ and $T\in (0,\hat T)$ 
this shows that $\lim_{s\to 0^+} \zeta(s) = \sigma(0,0) = 0$.
From \ref{item:SP-ISS-VSR_lemma_2_bound}), 
it follows that $|\bar F(x,e,T)| \le C(s,s)$ for all $|x|\le s$, $|e|\le s$ and
$T\in (0,\check{T}(s,s))$ for every $s\ge 0$. 
From (\ref{eq:barT})--(\ref{eq:1111}) and the fact that
$C(\cdot,\cdot)$ is nondecreasing in each variable, it follows that
$\zeta(s) \le C(s,s) \leq C(\bar s,\bar s)$ for all $s\in [0,\bar s]$ with $\bar s>0$.
Thus, $\zeta(s)< \infty$ for every $s \ge 0$.
We have $\zeta : \R_{\geq0} \rightarrow \R_{\geq0}$, $\zeta(0)=0$, $\zeta$ is right-continuous at zero and,
by  \eqref{eq:X1}--\eqref{eq:1111}, nondecreasing. Then, by Lemma~2.5 of \cite{clarke1998asymptotic},
there exists $\chi \in \K_\infty$ such that $\chi \geq \zeta$.\mer

Define $\eta \in \K_\infty$ and $\tilde \eta \in \K_\infty$ via
\begin{align}
  \label{eq:defEta}
  &\eta := \max \{ \chi , \text{id} \}, \\
  &\tilde  \eta := \max \{ \chi ,\rho \}.
\end{align}

By \eqref{eq:bounding} we have 
\begin{align}
  \sigma(s,r) &\leq   \tilde \eta (s)+\eta(r) \quad \forall s\geq0, \forall r\geq0.
\end{align}

Consider $M_0>0$, $E_0>0$ and $R_0>0$ given and $\alpha_1,\alpha_2,\alpha_3 \in \Kinf$ from \ref{item:SP-ISS-VSR_lemma_4}). 
Select $E=E_0$, $R = \eta^{-1}( \frac{1}{2} \alpha_2^{-1}( \frac{1}{3} \alpha_1(R_0)))$,
and $M = \alpha_1^{-1} \comp \alpha_2\big( \max\{R, M_0, \tilde \eta (E_0)+\eta(R_0)\} \big)$. From (\ref{eq:SP-ISS-VSR_d_1}), 
it follows that $M\ge R$. Also, $R\le R_0$ follows because, 
from (\ref{eq:SP-ISS-VSR_d_1}), $\frac{1}{2}\alpha_2^{-1}(\frac{1}{3}\alpha_1(r)) \le \alpha_2^{-1}(\alpha_1(r)) \le r$ for all $r\ge 0$,
and since $\eta(r) \ge r$ from (\ref{eq:defEta}), then $\eta^{-1}(r) \le r$ for all $r\ge 0$.
Let $M,E,R$ with \ref{item:SP-ISS-VSR_lemma_4}) generate $\tilde T>0$ and $V:\R^n \rightarrow \R_{\geq 0}$ such that (\ref{eq:SP-ISS-VSR_d_1}) and (\ref{eq:SP-ISS-VSR_d_2}) hold. 
Define $T^\star=\min\{\tilde T, \bar T(E,R)\}$ and
\begin{align}
  \X_2(s,r)&:= \{x: V(x)\leq \alpha_2(\tilde \eta(s)+\eta(r)) \}. \label{eq:X2}
\end{align}
Let $x\in\X_1(s,r)$. Then, $\alpha_2(|x|) \le \alpha_2(\rho(s) + r) \le \alpha_2(\tilde \eta(s) + \eta(r))$, 
and using (\ref{eq:SP-ISS-VSR_d_1}), then $V(x) \le \alpha_2 (\tilde \eta(s) + \eta(r))$. Therefore, $\X_1(s,r)\subset \X_2(s,r)$ for all $s\ge 0$ and  $r\ge 0$. Let $x_k$ denote the solution to \eqref{eq:system1} corresponding to $|x_0|\leq M_0$, $\|\{e_i\}\| \le E_0$ and $\{T_i\}\in \Phi(T^\star)$. From (\ref{eq:SP-ISS-VSR_d_1}) we have that $\alpha_2^{-1}(V(x_{k}))\leq |x_k|$; using this in (\ref{eq:SP-ISS-VSR_d_2}) then
\begin{align}
  \label{eq:diffVkp2} 
  V(x_{k+1}) &-V(x_{k}) \leq -T_k\alpha_3 (|x_k|) \le -T_k \alpha(V(x_{k})) \notag \\
             & \quad \text{if} \ms\ms \rho(|e_k|)+R \leq |x_k|\leq M
\end{align}
where $\alpha \dfn \alpha_3 \comp \alpha_2^{-1}$. 

\begin{claim}
\label{clm:bounded}
If $|x_0| \le M_0$ then $|x_k| \le M$ for all $k\in\N_0$.
\end{claim}

\indent\emph{Proof of Claim~\ref{clm:bounded}:}
From (\ref{eq:SP-ISS-VSR_d_1}), $|x_0| \le M_0$ implies that $V(x_0) \le \alpha_2(M_0)$ and $|x_0| \le \alpha_1^{-1} \comp \alpha_2(M_0)$, and by definition of $M$, then $|x_0| \le M$. By induction, we will prove that $V(x_k) \le \alpha_2( \max\{ M_0, \tilde\eta(E_0) + \eta(R_0) \})$ for all $k\in\N_0$. Note that the assertion holds for $k=0$. Suppose that $V(x_k) \le \alpha_2( \max\{ M_0, \tilde\eta(E_0) + \eta(R_0) \})$. 
Then, $|x_k| \le M$. If $x_k \notin \X_1(|e_k|,R)$, then $|x_k| >  \rho(|e_k|)+R$ and from (\ref{eq:diffVkp2}), then $V(x_{k+1}) \le V(x_k)$. If $x_k \in \X_1(|e_k|,R)$, 
from \eqref{eq:1111}--\eqref{eq:bounding}
and the definition
of $\tilde \eta$ and $\eta$ we have $|x_{k+1}| \leq \tilde\eta(|e_k|)+\eta(R)$.
Using (\ref{eq:SP-ISS-VSR_d_1}), then $V(x_{k+1}) \le \alpha_2(\tilde\eta(|e_k|)+\eta(R)) \le \alpha_2(\tilde\eta(E_0)+\eta(R_0))$, and hence the induction assumption holds for $k+1$. Since $V(x_k) \le \alpha_2( \max\{ M_0, \tilde\eta(E_0) + \eta(R_0) \})$ implies that $|x_k| \le M$, we have thus shown that $|x_k| \le M$ for all $k\in\N_0$.\mer

\begin{claim}
\label{clm:sets}
If $x_{\ell}\in \X_2(\|\{e_i\}\|,R)$ for some $\ell\in \N_0$ then $x_k$ remains in $\X_2(\|\{e_i\}\|,R)$ for all $k\geq \ell$.
\end{claim}

\indent\emph{Proof of Claim~\ref{clm:sets}:}
Let $x_{\ell} \in \X_2(\|\{e_i\}\|,R)$.
If $x_{\ell} \notin \X_1(\|\{e_i\}\|,R)$, then $|x_\ell| >  \rho(\|\{e_i\}\|)+R$. Consequently, if $x_{\ell} \in \X_2(\|\{e_i\}\|,R) \setminus \X_1(\|\{e_i\}\|,R)$, from \eqref{eq:diffVkp2} it follows that
\begin{align}
V(x_{\ell+1}) &\leq V(x_{\ell})-T_{\ell} \alpha(V(x_{\ell}))  \notag 
 \leq V(x_{\ell})  
\end{align}
and hence $x_{\ell+1} \in \X_2(\|\{e_i\}\|,R)$. 
Next, consider that  $x_{\ell} \in \X_1(\|\{e_i\}\|,R)$.
From \eqref{eq:1111} and the definition
of $\tilde \eta$ and $\eta$ we have $|x_{\ell+1}| \leq \tilde \eta(\|\{e_i\}\|)+\eta(R)$.
Using (\ref{eq:SP-ISS-VSR_d_1}) and recalling \eqref{eq:X2}, then $x_{\ell+1} \in \X_2(\|\{e_i\}\|,R)$. By induction, we have thus shown that if $x_\ell \in \X_2(\|\{e_i\}\|,R)$ for some $\ell\in\N_0$, then $x_k \in \X_2(\|\{e_i\}\|,R)$ for all $k\ge \ell$.\mer

Define $t_k=\sum_{i=0}^{k-1} T_i$ and the function
\begin{align*}
  y(t):=V(x_{k})+ &\frac{t-t_k}{T_{k}} \left(V(x_{k+1})-V(x_{k}) \right) 
                  &\text{if } t \in \left[t_k,t_{k+1}\right),
\end{align*}
which depends on the initial condition $x_0$, on the sampling period sequence
$\{T_i\}$, on the disturbance sequence $\{e_i\}$, and on the given constants $M_0,E_0,R_0$ (through the fact that $V$ depends on the latter constants).
Then, 
\begin{align}
\label{eq:doty} 
\dot{y}(t)= \frac{V(x_{k+1})-V(x_k)}{T_k}  \quad \forall t \in \left(t_k,t_{k+1}\right), \forall k \in \N_0.
\end{align}
Note that
\begin{align}
\label{eq:yProperty11}
y(t_k) &= V(x_{k}), \quad \forall k \in \N_0.
\end{align}
By Claim~\ref{clm:bounded},
for all $|x_0|\leq M_0$ 
and all $t_k$ such that
$x_k \notin \X_2(\|\{e_i\}\|,R)$,
\eqref{eq:diffVkp2} holds. Combining
\eqref{eq:diffVkp2} with 
\eqref{eq:doty}  it follows that for all $t\in(t_{k},t_{k+1})$
\begin{align}
\label{eq:yProperty2}
\dot{y}(t) &\leq -\alpha \left(y(t_k) \right) 
\leq -\alpha\left(y(t) \right).
\end{align}
Hence (\ref{eq:yProperty2}) holds up to $t=t_{k^*}$ where $t_{k^*}=\inf \{t_k: x_k \in \X_2(\|\{e_i\}\|,R) \}$.
Note that the function $\alpha = \alpha_3 \comp \alpha_2^{-1}$ does not depend on any of the following quantities: $x_0$, $\{T_i\}$, $\{e_i\}$, $M_0$, $E_0$, or $R_0$. Since $\alpha$ is positive definite, using Lemma~4.4 of \cite{LinSontagWangSIAM96}, then there exists $\beta_1 \in \KL$ such that,
for all $t\in[0,t_{k^*})$ we have
\begin{align*}
y(t) \leq \beta_1\left(y(0),t\right). 
\end{align*}
By \eqref{eq:yProperty11}, for every $k\in \N_0$ such that $x_j \notin \X_2(\|\{e_i\}\|,R)$ for all $0\leq j \leq k$,
we have that
\begin{align}
\label{eq:54}
 V(x_k) \leq \beta_1\left(V(x_0),\sum_{i=0}^{k-1} T_i\right).
\end{align}
From Claim~\ref{clm:sets} and \eqref{eq:X2}, if $x_{j} \in  \X_2(\|\{e_i\}\|,R)$ 
then $V(x_j)\leq \alpha_2(\tilde \eta(\|\{e_i\}\|)+\eta(R))$ for all $k\geq j$.
Combining the latter with \eqref{eq:54}, then, for all $k\in\N_0$,
\begin{align*}
 V(x_k) &\leq \beta_1\left(V(x_0),\sum_{i=0}^{k-1} T_i\right)+  \alpha_2\Big(\tilde \eta(\|\{e_i\}\|)+\eta(R)\Big).
\end{align*}
Using the fact that $\chi(a+b) \le \chi(2a)+\chi(2b)$ for every $\chi\in\K$
and (\ref{eq:SP-ISS-VSR_d_1}) then
\begin{align*}
  \alpha_1(|x_k|) &\leq \beta_1\left(\alpha_2(|x_0|),\sum_{i=0}^{k-1} T_i\right)+   \alpha_2(2\tilde \eta(\|\{e_i\}\|))+ \alpha_2(2\eta(R)).
\end{align*}
Define $\beta \in \KL$ via
$\beta(s,\tau):=\alpha_1^{-1}(3\beta_1(\alpha_2(s),\tau))$
and $\gamma \in \K_\infty$ via $\gamma(s):= \alpha^{-1}_1(3\alpha_2(2\tilde \eta(s)))$.
Recalling the definition of $R$ 
then $\alpha^{-1}_1(3\alpha_2(2\eta(R)))\leq R_0$, it follows that
\begin{align}
|x_k| &\leq 
   \beta\left(|x_0|,\sum_{i=0}^{k-1} T_i\right)+\gamma(\|\{e_i\}\|)+R_0
\end{align}
for all $k \in \N_0$, all $\{T_i\}\in \Phi(T^\star)$, all $|x_0|\leq M_0$ and all $\|\{e_i\}\|\leq E_0$. 
We have thus established that \eqref{eq:system1} is SP-ISS-VSR.
\end{IEEEproof}

\section{Examples}
\label{sec:example}

Consider the nonlinear continuous-time plant in Example~1 of \cite{NesicSCL99}:
\begin{align}
  \label{eq:cont_plant}
  \dot x &= x^3+u =: f(x,u)
\end{align}
whose Euler (approximate) discrete-time model is 
\begin{align}
\label{eq:discrete_aprox_plant}
  x_{k+1} &= x_k+T_k \left ( x_k^3+ u_k\right ) =: F^a(x_k,u_k,T_k).
\end{align}
This open-loop Euler model $F^a$ is consistent (as per Definition~\ref{def:consistent}) with the open-loop exact model $F^e$. This can be established by means of Lemma~\ref{lemma:sufcondfor_consistency_f} as follows.
Since $F^a$ coincides with $F^{\text{Euler}}$, then assumption~\ref{item:sufcondconsistency_consistentEuler}) of
Lemma~\ref{lemma:sufcondfor_consistency_f} holds; also, assumption~\ref{item:fbounded}) of Lemma~\ref{lemma:sufcondfor_consistency_f} is easily shown to hold using $f$ in (\ref{eq:cont_plant}).

We will consider two of the feedback laws considered in \cite{NesicSCL99}:
\begin{align}
  \label{eq:ex_controllaw}
  U(x_k,T_k) &= - x_k - 3x_k^3,\\
  \label{eq:ex_controllaw2}
  W(x_k,T_k) &= - x_k - x_k^3.
\end{align}
In \cite{NesicSCL99}, both control laws were shown to achieve semiglobal-practical stabilization under zero-order hold and uniform sampling. 
Our aim here is to show that when state measurement errors are taken into consideration, the control law (\ref{eq:ex_controllaw}) achieves SP-ISS-VSR whereas (\ref{eq:ex_controllaw2}) does not, not even under uniform sampling.

\subsection{SP-ISS-VSR}

Under the feedback law (\ref{eq:ex_controllaw}) and taking measurement errors into account so that $u_k = U(x_k+e_k,T_k)$, the closed-loop approximate model is given by (\ref{eq:system1}), with
\begin{align}
  \label{eq:Fa_example}
  \bar F^a(x,e,T)= x - T [2 x^3 + 9ex^2 + (9e^2 + 1)x + 3e^3 + e].
\end{align}
We next prove that $(U,F^a)$ is MSEC with $(U,F^e)$ by means of Lemma~\ref{lemma:sufcondconsistency}. Assumption~\ref{item:sufcondconsistency_consistentEuler}) has been already established, whereas \ref{item:sufcondconsistency_controllocbound}) is easily shown using (\ref{eq:ex_controllaw}). In order to establish \ref{item:sufcondconsistency_Fa_nondecreasing_bound}), we use (\ref{eq:Fa_example}) to evaluate
\begin{align*}
  |\bar F^a(x,e,T)-\bar F^a(z,e,T)| &= 
   | 1 - T P(x,e,z) |\ |x-z|\\
  &\le (1 + T |P(x,e,z)|) |x-z|,
\end{align*}
where $P(x,e,z)$ denotes a multinomial in the indeterminates $x,e,z$, and we have used the fact that for every positive integer $p$, $x^p -z^p = q(x,z)\cdot(x-z)$ for some multinomial $q(x,z)$. Consider compact sets $\X,\E\subset\R$, and define the nonnegative constant $\bar\sigma := \sup_{(x,e,z) \in (\X,\E,\X)} |P(x,e,z)|$. Then,
\eqref{eq:lemma_sufcondconsistency_Fa_nondecreasing} holds for all $T>0$ with $\sigma : \R_{\ge 0} \to \R_{\ge 0}$ defined as $\sigma(T) := \bar\sigma$, and assumption~\ref{item:sufcondconsistency_Fa_nondecreasing_bound}) of Lemma~\ref{lemma:sufcondconsistency} holds. By Lemma~\ref{lemma:sufcondconsistency}, $(U,F^a)$ is MSEC with $(U,F^e)$.

Next, we prove that $x_{k+1}=\bar F^a(x_k,e_k,T_k)$ is SP-ISS-VSR.
The continuity and boundedness assumptions
\ref{item:SP-ISS-VSR_lemma_3_zero}), \ref{item:SP-ISS-VSR_lemma_1_cont}) and \ref{item:SP-ISS-VSR_lemma_2_bound})
of Theorem~\ref{lemma:SP-ISS-VSR_Lyap}, 
can be easily verified from (\ref{eq:discrete_aprox_plant}),
(\ref{eq:ex_controllaw}) and (\ref{eq:Fa_example}). 
To prove assumption~\ref{item:SP-ISS-VSR_d}) of Theorem~\ref{lemma:SP-ISS-VSR_Lyap}, let $\alpha_1,\alpha_2,\alpha_3,\rho \in \Kinf$ be defined via $\alpha_1(s) =\alpha_2(s) = s^2$, $\alpha_3(s) = 2.194s^4$, and $\rho(s) = 10 s$. Let $M\ge R>0$ and $E>0$ be given and define $V(x)=x^2$. Note that $V$ satisfies (\ref{eq:SP-ISS-VSR_d_1}) and compute
\begin{align}
  \label{eq:Lyap_1}
  &V(\bar F^a(x,e,T))-V(x) = h(x,e)T+g(x,e)T^2, \\
  &h(x,e) = -2x[2 x^3 + 9ex^2 + (9e^2 + 1)x + (3e^3 + e)],\\
  &g(x,e) = [2 x^3 + 9ex^2 + (9e^2 + 1)x + (3e^3 + e)]^2.
\end{align}
Define $G := \sup_{|x|\le M,|e|\le E} g(x,e)$ and let $\tilde T = \frac{1.8R^2}{G}$.
Expanding $h(x,e)$ and taking absolute values on sign indefinite terms,
\begin{align}
  h(x,e) &\leq - 4x^4 -(18e^2 + 2) x^2 + 18|e||x|^3  + ( 6|e|^3 + 2|e|)|x| \notag\\
  \label{eq:h(x,e)}
  &\le - 4x^4 - 2 x^2 + 18|e||x|^3  + ( 6|e|^3 + 2|e|)|x|.
\end{align}
Whenever $\rho(|e|) \le |x|$, it follows that $|e| \le 0.1|x|$ and
\begin{align*}
  h(x,e) &\leq -4x^4 - 2 x^2 + 1.8 x^4 + (0.006 |x|^3 + 0.2 |x|)|x| \\
  &\le -2.194 x^4 - 1.8 x^2.
\end{align*}
It thus follows that for $\rho(|e|) + R \le |x| \le M$ and $T\in (0,\tilde T)$,
\begin{align}
  V(\bar F^a(x,e,T))-V(x) &\leq (- 2.194 x^4 - 1.8 x^2)T + GT^2 \notag \\
                          &\leq -T 2.194 x^4 + \left(- 1.8 x^2 +GT\right)T \notag\\
                          &\le -T \alpha_3(|x|).
\end{align}
Therefore, assumption~\ref{item:SP-ISS-VSR_d}) of Theorem~\ref{lemma:SP-ISS-VSR_Lyap} also holds
and the closed-loop system $x_{k+1} = \bar F^a(x_k,e_k,T_k)$ is SP-ISS-VSR. Theorem~\ref{theorem:ISS_aprox_to_exact} then ensures that the exact closed-loop system is SP-ISS-VSR.

\subsection{Practical stability but no SP-ISS-VSR}

Consider next the feedback law (\ref{eq:ex_controllaw2}), which was also shown in \cite{NesicSCL99} to achieve semiglobal practical stability under zero-order hold and uniform sampling. We next show that, under bounded state measurement errors and also uniform sampling, the true plant state may diverge. Under the feedback law (\ref{eq:ex_controllaw2}), the closed-loop approximate model becomes
\begin{align}
  \label{eq:discrete_aprox_plant_2}
  \bar F^a(x,e,T) &= x - T [3ex^2 + (1+3e^2)x + e^3+e].
\end{align}
Notice the absence of the cubic term in $x$ within the square brackets in (\ref{eq:discrete_aprox_plant_2}) as compared with (\ref{eq:Fa_example}).
Consider the constant error sequence $\{e_i\}$ with $e_i = -1$ for all $i\in\N_0$.
From \eqref{eq:discrete_aprox_plant_2}, then
\begin{align*}
  \bar F^a(x,-1,T) &= x + T [3x^2 -4x + 2].
\end{align*}
The polynomial between square brackets satisfies $3x^2 -4x + 2 \ge x$ for all $x \ge 1$. Therefore, $\bar F^a(x,-1,T) \ge x (1+T)$ whenever $x\ge 1$ and $T>0$. For every $T^\star > 0$, consider the constant sequence $\{T_i\} \in \Phi(T^\star)$ with $T_i = T^\star/2$ for all $i\in\N_0$. It it thus clear that for the selected constant sequences $\{e_i\}$ and $\{T_i\}$, and for $x_0 \ge 1$, we have
\begin{align*}
  \lim_{k\to\infty} x^a(k,x_0,\{e_i\},\{T_i\}) = \infty
\end{align*}
and hence the approximate closed-loop system $x_{k+1} = \bar F^a(x_k,e_k,T_k)$ is not SP-ISS-VSR. However, the pair $(W,F^a)$ is indeed MSEC with $(W,F^e)$, as can be shown following identical steps to those in the previous example. Using Lemma~\ref{LEMA3}, it can be shown that the exact closed-loop system cannot be SP-ISS-VSR, either.

\section{Conclusions}
\label{sec:conclusions}

We have given stability results for digital control design based on discrete-time approximate
models under varying sampling rate (VSR) and in the presence of state measurement errors.
We have extended the concept of semiglobal practical ISS (SP-ISS) to the VSR case (SP-ISS-VSR) 
and introduced the concept of multi-step error consistency (MSEC). 
We have shown that if the approximate closed-loop model is MSEC with the exact one, 
and if the control law renders the approximate model SP-ISS-VSR, then the same controller
ensures SP-ISS-VSR of the exact discrete-time closed-loop model. We have also given sufficient
conditions for MSEC and derived Lyapunov-based conditions that guarantee SP-ISS-VSR of a discrete-time model.
All of the given conditions are checkable without assuming knowledge of the exact discrete-time model.
We have also illustrated application via numerical examples.

\appendix

\begin{IEEEproof}[Proof of Lemma~\ref{lemma:sufcondfor_consistency_f}]
  Let $\Omega \subset \R^n \times \R^m$ be a given compact set. Let $\X_1 \subset \R^n$ and $\U \subset \R^m$ be compact sets such that $\Omega\subset\X_1 \times \U$, and define $\X:=N(\X_1,1)$. Let $\rho' \in \K$ and $T_0' > 0$ be given by Definition~\ref{def:consistent} in correspondence with $\Omega$ by the fact that $F^a$ is consistent with $F^{\text{Euler}}$. Let hypothesis \ref{subitem:f_bounded_a}) generate $\rho \in \K$ and $M>0$ in correspondence with $\X$ and $\U$. Define $T_0 := \min\{T_0',1/M\}$. Let $(x_0,u)\in\Omega$ and let $\phi_u(t,x_0)$ denote the (unique) solution to 
\begin{align*}
  \dot{x}(t) &= f(x(t),u), \quad x(0)=x_0, 
\end{align*}
where $u$ is constant, so that
\begin{align*}
  \phi_u(t,x_0) &= x_0 + \int_0^t f(\phi_u(s,x_0),u) ds.
\end{align*}
From hypothesis \ref{subitem:f_bounded_a}), then $|\phi_u(t,x_0) - x_0|\leq M t\le 1$, and hence $x(t)\in \X$, holds for all $t\in [0,T_0]$. From hypothesis \ref{subitem:f_Lypschitz_b}), it follows that for all $T\in [0,T_0]$,
\begin{multline*}
  \left|\int_0^T \left[f(\phi_u(t,x_0),u)-f(x_0,u)\right]dt \right|\\ \leq \int_0^T \rho(|\phi_u(t,x_0)-x_0|)dt \leq T \rho(MT). 
\end{multline*}
Considering that $F^e(x_0,u,T) = x_0 + \int_0^T f(\phi_u(t,x_0),u) dt$, that $F^{\text{Euler}}(x_0,u,T) = x_0 + \int_0^T f(x_0,u) dt$, and that $F^a$ is consistent with $F^{Euler}$, then, for all $T\in [0,T_0]$, it follows that
\begin{align*}
|F^e(x_0,u,T)&-F^a(x_0,u,T)| \le |F^{\text{Euler}}(x_0,u,T) - F^a(x_0,u,T)|\\
  &\hspace{5mm} + |F^e(x_0,u,T) - F^{\text{Euler}}(x_0,u,T)|\\
&\le T \rho'(T) +\left|\int_0^T \left [ f(\phi_u(t,x_0),u)-f(x_0,u) \right]dt\right| \\
&\leq T\rho'(T)+T\rho(MT) =: T\rho_1(T)
\end{align*}
where $\rho_1 \in \K$. Thus, $F^e$ is consistent with $F^a$.
\end{IEEEproof}

\begin{IEEEproof}[Proof of Lemma~\ref{LEMA3}]
  Define $\tilde \X:=N(\X,\eta)$. Since $(U,F^a)$ is MSEC with $(U,F^e)$, in correspondence with the tuple $(\tilde \X,\E,L,\eta)$ there exist $\alpha : \R_{\geq0} \times \R_{\geq0} \rightarrow \R_{\geq0}$  and $T^*>0$ such that
\eqref{eq:mscee1} and \eqref{eq:mscee2} hold according to Definition~\ref{def:MSEC}. Let $\tilde T := T^*$ and consider $\{T_i\} \in \Phi(\tilde T)$, $\{e_i\} \subset \E$, and $\xi \in \X$. Note that (\ref{eq:MSEC2}) holds trivially for $k=0$ because
\begin{align*}
  |\Delta x_0| &= |\xi-\xi| = 0 \le \eta.
\end{align*}
We proceed by induction on $k$. Let $k\ge 0$ be such that $\sum_{i=0}^{k} T_i \in [0,L]$. Note then that $\sum_{i=0}^{j} T_i \in [0,L]$ for every $0 \le j \le k$. Suppose also that $x^a(j,\xi,\{e_i\},\{T_i\}) \in \X$ and that $|\Delta x_j|\leq \eta$, both for all $0 \le j \le k$. Then, $x^e(j,\xi,\{e_i\},\{T_i\}) \in \tilde\X$ for all $0 \le j \le k$. Consider $\Delta x_{k+1}$. From (\ref{eq:MSEC2}), (\ref{eq:system1}) and (\ref{eq:mscee1}), we have
\begin{align*}
  |\Delta x_{k+1}| &= \left|\bar F^e\big(x^e_k,e_k,T_k\big) - \bar F^a\big(x^a_k,e_k,T_k\big)\right| \le \alpha(|\Delta x_k|,T_k) \\ &\le \alpha^{k+1}(|\Delta x_0|,\{T_i\}) = \alpha^{k+1}(0,\{T_i\}) \le \eta.
\end{align*}
We have thus established the result by induction.
\end{IEEEproof}

\bibliographystyle{IEEEtran}
\bibliography{bibliografia}
\end{document}